\documentclass[11pt,a4paper]{article}
\usepackage[utf8]{inputenc}
\usepackage[T1]{fontenc}
\usepackage[english]{babel}
\usepackage{lmodern}
\usepackage{microtype}
\usepackage[a4paper,top=2.5cm,bottom=2.5cm,left=2.7cm,right=2.7cm,footskip=1.2cm]{geometry}
\usepackage{amsmath,amssymb,amsthm,mathtools,bm}
\usepackage{graphicx,booktabs,array,tabularx,multirow,makecell,longtable}
\usepackage{caption,subcaption}
\usepackage{xcolor,enumitem,csquotes}
\usepackage{tikz}
\usetikzlibrary{arrows.meta,positioning,fit,calc,shapes.geometric,decorations.pathreplacing}
\usepackage{siunitx}
\newcolumntype{L}[1]{>{\raggedright\arraybackslash}p{#1}}
\newcolumntype{Y}{>{\raggedright\arraybackslash}X}
\usepackage[backend=biber,style=authoryear-comp,natbib=true,url=true,doi=true,isbn=false,maxbibnames=99,maxcitenames=2,uniquelist=false,sorting=nyt]{biblatex}
\usepackage[colorlinks=true,linkcolor=black,citecolor=black,urlcolor=black]{hyperref}
\usepackage[noabbrev,capitalize]{cleveref}
\theoremstyle{plain}
\newtheorem{proposition}{Proposition}

\newtheorem{lemma}[proposition]{Lemma}
\newtheorem{corollary}[proposition]{Corollary}
\theoremstyle{definition}
\newtheorem{definition}{Definition}
\newtheorem{assumption}{Assumption}

\newtheorem{criterion}{Criterion}
\theoremstyle{remark}

\crefname{assumption}{Assumption}{Assumptions}
\Crefname{assumption}{Assumption}{Assumptions}
\crefname{definition}{Definition}{Definitions}
\Crefname{definition}{Definition}{Definitions}
\crefname{criterion}{Criterion}{Criteria}
\Crefname{criterion}{Criterion}{Criteria}
\setlist{nosep,leftmargin=*}
\definecolor{seriesblue}{RGB}{31,78,121}
\definecolor{seriesgray}{RGB}{242,244,247}
\tikzset{seriesbox/.style={draw=black!70,rounded corners=2pt,align=center,inner sep=5pt,minimum height=9mm,fill=seriesgray},seriesarrow/.style={-{Latex[length=2.2mm]},thick,draw=black!75}}

\hypersetup{pdftitle={Manipulation, Informed Trading, and Regulation in Leveraged Event-Linked Markets},pdfauthor={Maksym Nechepurenko},pdfsubject={Event-linked leverage; manipulation; informed trading; regulation},pdfkeywords={prediction markets, event contracts, leverage, manipulation, informed trading, regulation}}
\title{Manipulation, Informed Trading, and Regulation\\in Leveraged Event-Linked Markets}
\author{Maksym Nechepurenko\thanks{Founder and Director of Research, ForesightFlow, the Research Department of Devnull FZCO, Dubai, United Arab Emirates. Email: \texttt{maksym@devnull.ae}.}}
\date{July 19, 2026}
\begin{document}
\maketitle
\begin{abstract}
Leverage does not create manipulation or informed trading in event markets, but it changes their economics. This paper separates four conduct channels that are often conflated: trading that distorts the quoted market price, conduct that changes the real-world event, interference with the resolution process, and informed trading that exploits non-public information without changing either the event or the resolution rule. We model each channel at the level needed for market design and regulation.

The central analytical result is a capital-constrained amplification identity. When gross event exposure is financed at leverage $L$, the return to a successful directional advantage scales with gross exposure while several important costs---information acquisition, operational influence, concealment, and enforcement exposure---need not scale proportionally. The resulting threshold is not a universal ``safe leverage'' number. It depends on executable capacity, the probability and magnitude of influence, detection and sanction functions, position limits, and whether profits can be realized before finality. We prove that leverage weakly enlarges the set of profitable influence opportunities under fixed-cost influence, but that the conclusion can be reversed by endogenous price impact, convex detection, or enforceable position limits.

The paper then connects the incentive model to the revised findings of Papers 1, 2, and 4 in this series. Historical-path replay is evidence about mechanical rule channels, not equilibrium conduct. Fill-side address data can support concentration and trading-intensity analysis, but Polymarket's off-chain quote lifecycle prevents address-level claims about spoofing or posted-liquidity withdrawal. Accordingly, the surveillance design is evidence-layered: fills, attributed quotes, external-event integrity, and oracle governance are treated as separate data domains.

The regulatory contribution is a principles-first framework rather than a legal classification claim. A leveraged event venue should be evaluated along five functions: product authorization, market-integrity surveillance, event-influence controls, resolution governance, and loss-bearing capital. A July 2026 source snapshot shows an actively evolving United States framework and no harmonized cross-jurisdictional category for leveraged event contracts. The paper therefore treats legal status as jurisdiction- and product-specific, while deriving venue requirements that remain valid across classifications.
\end{abstract}

\clearpage
\section{Introduction}
A trader can profit from an event contract in at least four economically different ways. The trader may forecast the event better than others. The trader may possess non-public information. The trader may move the quoted price and exit before correction. Or the trader may influence the event or the resolution process itself. A leverage layer changes all four, but not in the same way.

The distinction matters because a market can appear manipulation-resistant under one definition and remain exposed under another. A robust index can make short-lived mark manipulation more expensive without reducing the incentive to bribe an event participant. A strict oracle dispute process can protect final settlement without preventing trading on confidential information. A trading halt can remove one execution window while creating a focal time for pre-halt positioning. Regulation that treats all conduct as generic ``market abuse'' misses these mechanism differences.

This paper develops a compact framework for the conduct and regulatory problems created by leveraged event-linked markets. It revises the original Paper 3 in three respects. First, it replaces unsupported class-specific numerical scenarios with propositions and comparative statics. Second, it aligns all claims with the corrected balance-sheet logic of Paper 1 and the orthogonal design axes of Paper 2. Third, it removes stale or overconfident jurisdictional statements and uses a dated primary-source regulatory snapshot.

\subsection{Research questions}
We ask four questions:
\begin{enumerate}[label=(\roman*)]
\item Which forms of manipulation are specific to event-linked instruments, and which are ordinary trading-venue conduct applied to a bounded claim?
\item Under what conditions does leverage enlarge the profitable set of event influence or informed trading?
\item What can historical fill-side data identify, and what remains unobservable without quote-lifecycle, identity, event-integrity, or oracle-governance data?
\item Which control functions should apply regardless of whether a jurisdiction classifies the product as a derivative, event contract, crypto-asset instrument, or gambling product?
\end{enumerate}

\subsection{Contributions}
The paper contributes: (1) a four-channel conduct taxonomy; (2) a capital-constrained leverage-amplification model with explicit countervailing forces; (3) a separation of manipulation, informed trading, and lawful influence; (4) an evidence-layered surveillance architecture; and (5) a functional regulatory framework tied to product authorization, surveillance, event integrity, resolution governance, and loss-bearing capital.

\section{Instrument and balance-sheet setting}
Let an event claim pay $Y\in\{0,1\}$ at finality. A trader enters a directional position of gross exposure $N$ at price $p\in(0,1)$. The trader contributes capital $C>0$ and receives leverage
\[
L=\frac{N}{C}\ge 1.
\]
The paper does not assume that every leveraged event instrument is a conventional perpetual. The same conduct analysis applies to a synthetic perpetual, a dated margined future, or a physically backed financed claim, subject to differences in liquidation and finality.

For a long exposure held to finality, gross trading profit before fees is
\[
\Pi^{\mathrm{trade}}=N(Y-p).
\]
This identity is distinct from the account equity and creditor-loss identities in Paper 1. It measures the trader's directional payoff; it does not by itself establish that the position survives liquidation, that the venue can collect losses, or that a lender is protected.

\begin{assumption}[Conduct opportunity]
A conduct opportunity is described by an action $a\in\mathcal A$ that may affect one or more of: the trading price path, the distribution of $Y$, the reported resolution, or the trader's information set.
\end{assumption}

Let $K(a)$ be the direct cost of action, $D(a,N)$ the expected detection-and-sanction cost, $X(a,N)$ execution and price-impact cost, and $G(a,N)$ the expected gross payoff advantage relative to honest uninformed trading. Net advantage is
\[
V(a,N)=G(a,N)-K(a)-D(a,N)-X(a,N).
\]
Leverage matters through $N=LC$, but the sign of $\partial V/\partial L$ depends on how every term scales.

\section{A four-channel conduct taxonomy}
\begin{definition}[Market-price manipulation]
Conduct that changes a venue price, index input, funding input, or liquidation trigger without changing the real-world event or the valid final resolution.
\end{definition}

\begin{definition}[Outcome manipulation]
Conduct intended to change the distribution or realization of the underlying real-world event.
\end{definition}

\begin{definition}[Resolution manipulation]
Conduct intended to change, delay, censor, or strategically exploit the mechanism that maps evidence about the event into the contract's final payout vector.
\end{definition}

\begin{definition}[Informed trading]
Trading based on information not incorporated in the public price, without the trader changing the event or corrupting the resolution mechanism. Whether the conduct is unlawful depends on duties, source, venue rules, and jurisdiction; informational advantage alone is not a legal conclusion.
\end{definition}

The channels differ by causal object, evidence, and remedy. Price manipulation is primarily a market-microstructure problem. Outcome manipulation is an event-integrity problem. Resolution manipulation is a governance and oracle-security problem. Informed trading is an information-rights and market-access problem. A single transaction sequence can participate in more than one channel, but the channels should not be merged analytically.

\begin{table}[ht]
\centering\small
\caption{Conduct channels and primary evidence.}
\begin{tabularx}{\textwidth}{@{}L{2.8cm}YYY@{}}
\toprule
Channel & Causal object & Primary evidence & Principal controls\\
\midrule
Market price & Quotes, trades, index inputs & Attributed order lifecycle, fills, book state & Index robustness, anti-spoof rules, position and order controls\\
Outcome & Real-world event & External integrity records, participant links, position timing & Restricted-person rules, event-body cooperation, limits\\
Resolution & Proposal, dispute, evidence, payout mapping & Oracle and governance logs, evidence provenance & Separation of duties, bonds, challenge rights, fallback rules\\
Informed trading & Information advantage & Identity/duty evidence plus trading pattern & Access controls, restricted lists, surveillance, reporting\\
\bottomrule
\end{tabularx}
\end{table}

\section{Market-price manipulation under bounded support}
A bounded claim changes the geometry but not the basic logic of trade-based manipulation. Let $m$ be the pre-action mark and let an action generate a temporary displacement $\delta$ that can be monetized on a pre-existing position $N$. A reduced-form benefit is $N\delta$, while execution cost is $X(\delta,N)$.

\begin{proposition}[No linear free amplification]
Suppose $X(\delta,N)$ is convex in $N$, $X(\delta,0)=0$, and $\partial X/\partial N\to\infty$ as $N$ approaches executable capacity. Then increasing leverage does not imply unbounded manipulation profit, even when gross benefit is linear in $N$.
\end{proposition}
\begin{proof}
Net benefit is $N\delta-X(\delta,N)-K-D$. Convexity and the divergent marginal execution cost imply a finite maximizer whenever one exists. Leverage expands the feasible notional set but does not remove the execution-capacity boundary.
\end{proof}

Bounded support creates direction-dependent room to move: an upward displacement cannot exceed $1-m$ and a downward displacement cannot exceed $m$. It does not imply that manipulation is easiest near the midpoint. That empirical conclusion depends on executable depth, spread, fees, and the index rule. Paper 4 shows why address-level spoofing claims are not available from on-chain fills alone: order placement and cancellation occur off-chain in the hybrid central limit order book. Fill-side wash or reversal candidates are descriptive signals, not proof of spoofing or intent.

Index manipulation and liquidation-trigger manipulation require separate treatment. If a venue uses the same thin external book both as trade venue and liquidation reference, an attacker may profit without exiting the manipulated instrument: the price displacement can trigger third-party liquidations or funding transfers. Robust index construction therefore needs source diversity, stale-source rules, depth and dispersion checks, and explicit behavior when no trustworthy source remains.

\section{Outcome manipulation and leverage}
Let action $a$ increase the probability of a favorable outcome from $q_0$ to $q_1(a)$, with $\Delta q(a)=q_1(a)-q_0\ge0$. For a long position entered at price $p$, the incremental expected gross benefit attributable to influence is
\[
G(a,N)=N\Delta q(a).
\]
Net influence value is
\[
V_I(a,L;C)=LC\Delta q(a)-K(a)-D(a,LC)-X(a,LC).
\]

\begin{proposition}[Leverage expands profitable influence under fixed costs]
Fix $a$ with $\Delta q(a)>0$. If $K(a)$ is independent of notional and $D(a,N)+X(a,N)$ is constant in $N$ over the feasible interval, then the set of $L$ for which $V_I(a,L;C)>0$ is an upper interval. Its threshold is
\[
L^*(a)=\frac{K(a)+D(a)+X(a)}{C\Delta q(a)}.
\]
\end{proposition}
\begin{proof}
Under the assumptions, $V_I$ is affine and strictly increasing in $L$. Solving $V_I>0$ gives the threshold.
\end{proof}

The proposition is intentionally narrow. It identifies the mechanism behind amplification but not a universal leverage cap. The fixed-cost assumption may approximate bribery, access acquisition, or operational interference over a limited notional range. It fails when the action must scale with event size, when detection rises with position visibility, or when the market price incorporates the influence risk.

\begin{proposition}[Countervailing convexity]
If $D(a,N)+X(a,N)$ is differentiable and its marginal cost exceeds $\Delta q(a)$ for all $N\ge \bar N$, then net influence value is decreasing beyond $\bar N$. Position limits at or below $\bar N$ can therefore prevent leverage from reaching a region of increasing net influence value.
\end{proposition}

Event classes should not be ranked by invented cost numbers. The correct empirical objects are: influence technology, event scale, access concentration, detection probability, sanction enforceability, executable position capacity, and linked-account aggregation. Sports, politics, corporate actions, and macro releases may differ on these dimensions, but the differences require event-specific evidence.

\section{Informed trading under leverage}
Let an informed trader estimate a favorable-outcome probability $q=p+e$, where $e>0$ is informational edge relative to the transaction price. Ignoring execution and liquidation, expected gross advantage is $Ne=LCe$.

Let $K_I$ denote information-acquisition and operational costs, and let $R(N)$ collect execution, financing, liquidation, and enforcement risks. The trader participates when
\[
LCe>K_I+R(LC).
\]

\begin{proposition}[Fixed-cost amortization]
If $K_I>0$ is fixed and $R(N)/N$ is non-increasing over a feasible interval, then the minimum edge required for participation,
\[
e^*(N)=\frac{K_I+R(N)}{N},
\]
is non-increasing in $N$.
\end{proposition}

Thus leverage can make weaker informational advantages privately tradable by spreading fixed acquisition costs across a larger position. This is different from claiming that Sharpe ratios are preserved. Financing cost, liquidation truncation, nonlinear price impact, and detection may change both expected return and variance.

The legal boundary is not statistical. A pattern of profitable pre-event trading can be consistent with research skill, lawful private information, breach of duty, prohibited participant influence, or outcome manipulation. Market surveillance can prioritize cases, but legal attribution requires evidence about identity, duty, source, and conduct.

\section{How risk-engine rules reallocate conduct incentives}
The revised Paper 1 finds that historical-path replay is a mechanical stress test, not an equilibrium deployment result. This distinction is especially important for manipulation. An engine rule may change the opportunity set that generated the observed path.

\subsection{Margin and leverage compression}
Higher margin reduces gross exposure for a fixed capital budget and therefore weakly reduces the linear amplification term. It can also force informed or manipulative traders to act earlier, split positions across accounts, or use correlated venues. A margin rule is therefore a capacity control, not a complete conduct control.

\subsection{Funding}
Funding transfers create a second monetization channel. An actor may manipulate the reference, the mark, or both to receive funding even without a favorable terminal outcome. Caps reduce payer insolvency but can weaken convergence incentives. Surveillance must attribute both price displacement and resulting transfer benefits.

\subsection{Trading halts}
A halt removes post-halt execution but creates a known boundary. It can concentrate position-taking before the cutoff and increase the value of information about whether the halt or close will occur as scheduled. A halt that mechanically reduces final-hour liquidations does not establish lower manipulation risk.

\subsection{Physical conversion and debt extinction}
A physically backed design that extinguishes debt before finality removes one creditor-loss channel, as developed in the Axient work, but it does not remove event influence, informed trading, oracle corruption, or execution manipulation during hard-flat. Instrument architecture reallocates the risk surface rather than eliminating it.

\section{Evidence-layered surveillance}
A credible surveillance system should state what each data layer can and cannot identify.

\begin{criterion}[Fill evidence]
On-chain fills support concentration, timing, size, counterparty, cross-market, and realized-position analyses. They do not reveal unfilled orders or cancellations.
\end{criterion}

\begin{criterion}[Quote-lifecycle evidence]
Spoofing, quote stuffing, cancellation races, and displayed-liquidity withdrawal require attributed order placement, modification, and cancellation records. Market-level snapshots without participant attribution are insufficient for address-level findings.
\end{criterion}

\begin{criterion}[External-event evidence]
Outcome manipulation requires records from the event domain: participant accreditation, integrity alerts, communications, event logs, and links between trading identities and persons with influence.
\end{criterion}

\begin{criterion}[Resolution evidence]
Oracle surveillance requires proposal, challenge, bond, evidence, voting, escalation, and final payout records with deterministic links to the traded market.
\end{criterion}

A surveillance alert should therefore carry an evidence grade. A fill anomaly can justify review but not a finding of spoofing. A profitable pre-event position can justify identity and duty checks but not a finding of insider trading. A disputed oracle request can justify governance review but not a finding that the real-world event was manipulated.

\section{Regulatory framework: functions before labels}
Legal classification is temporally unstable and jurisdiction-specific. This section therefore separates durable control functions from a dated source snapshot.

\subsection{Five regulatory functions}
\begin{enumerate}
\item \textbf{Product authorization.} Who may list the instrument, to whom, with what collateralization, leverage, disclosure, and suitability rules?
\item \textbf{Market integrity.} Which entity monitors manipulation, concentrated positions, linked accounts, funding transfers, and liquidation abuse?
\item \textbf{Event integrity.} Which persons are prohibited because they can influence the event, and how are sports bodies, employers, issuers, or public agencies connected to surveillance?
\item \textbf{Resolution governance.} Who controls definitions, evidence, proposal, dispute, override, void, and fallback decisions?
\item \textbf{Loss-bearing capital.} Who bears shortfall, operational loss, oracle error, custody failure, and unresolved disputes?
\end{enumerate}

\subsection{United States snapshot: July 2026}
The Commodity Futures Trading Commission (CFTC) describes regulated event contracts as derivatives, commonly swaps or futures, and subjects designated contract markets to core principles including resistance to manipulation and enforcement of contract terms \parencite{cftc_prediction_2026,cftc_anprm_2026}. In February 2026 the Commission withdrew its 2024 event-contract proposal; in March it opened a new advance notice of proposed rulemaking; and in June it proposed a structured framework for contracts involving enumerated activities and separate reporting rules for certain fully collateralized event contracts \parencite{cftc_withdrawal_2026,cftc_nprm_2026,cftc_reporting_2026}. The framework is therefore active and unsettled. The CFTC also issued 2026 advisories concerning non-public information and market obligations, illustrating that event-market integrity is not reducible to ordinary price manipulation \parencite{cftc_enforcement_2026,cftc_market_advisory_2026}.

Nothing in those sources establishes a general authorization for leveraged event-linked perpetuals. A product would require analysis of its exact contract form, trading venue, clearing, customer access, margin, and underlying event.

\subsection{Cross-jurisdictional interpretation}
Outside the United States, classification may involve derivatives, securities, crypto-asset, gambling, consumer-protection, or prohibited-offering regimes. The same user interface can embed different legal instruments. The paper therefore does not claim that a decentralized wrapper or offshore entity removes regulation. Nor does it claim that every event contract is gambling or every margined event instrument is a regulated future. Those are legal conclusions requiring current jurisdiction-specific analysis.

\begin{table}[ht]
\centering\small
\caption{Functional minimum independent of classification.}
\begin{tabularx}{\textwidth}{@{}L{3cm}YY@{}}
\toprule
Function & Minimum evidence & Minimum control\\
\midrule
Authorization & Instrument terms, customer class, collateral and leverage & Listing review, disclosures, limits\\
Market integrity & Orders, fills, positions, funding, liquidations & Real-time surveillance and enforcement\\
Event integrity & Restricted-person and event-body data & Access restrictions and cooperation protocols\\
Resolution & Versioned rules, evidence, proposals and disputes & Separation of duties and appeal/fallback path\\
Capital & Account and waterfall records & Segregation, stress tests, recovery plan\\
\bottomrule
\end{tabularx}
\end{table}

\section{Design implications}
The framework yields seven concrete requirements.

\begin{enumerate}
\item Position limits must aggregate linked accounts and correlated instruments, not only one wallet in one market.
\item Leverage tiers should depend on event influence concentration, executable capacity, and resolution governance, not only volatility.
\item Persons with direct or material influence over the event should be restricted or subject to enhanced disclosure and monitoring.
\item Quote-lifecycle records should be retained with participant attribution even when matching is off-chain.
\item Oracle and market-operation roles should be separated; emergency overrides should be versioned, logged, and reviewable.
\item Surveillance should trace benefit channels: trading profit, funding receipt, liquidation transfer, oracle reward, and related-market profit.
\item Loss controls and conduct controls should be evaluated separately. A solvent venue can remain manipulable; a well-surveilled venue can remain undercapitalized.
\end{enumerate}

No single control is sufficient. Low leverage cannot prevent corruption of a resolution committee. Strong oracle bonds cannot prevent an athlete from influencing a match. Know-your-customer controls cannot identify every beneficial owner or off-platform hedge. The defensible architecture is layered and explicit about residual risk.

\section{Limitations and empirical agenda}
The theoretical model is partial equilibrium. It does not solve for endogenous event prices, manipulator entry, market-maker withdrawal, cross-venue hedging, or strategic enforcement. The propositions isolate comparative statics under stated assumptions.

No new empirical calculation is performed. Paper 4's fill-side findings are used only to define observability boundaries. The paper does not estimate manipulation prevalence, insider-trading prevalence, detection probability, sanctions, influence costs, or safe leverage by event class.

The regulatory review is a dated research snapshot, not legal advice. Official sources can change after publication. Any deployment requires jurisdiction-specific counsel and a current review of product, venue, customer, custody, clearing, and marketing rules.

The most useful empirical extensions are: linked-entity position concentration; attributed quote-lifecycle analysis; event-participant and trader-link datasets under lawful access; oracle proposal and dispute reconstruction; cross-venue exposure; and quasi-experimental analysis of leverage or position-limit changes.

\section{Conclusion}
Leveraged event markets combine ordinary trading-venue abuse with two event-specific channels: influence over the real-world outcome and influence over the mapping from evidence to payout. Informed trading is a fourth channel that may resemble manipulation in transaction data but differs causally and legally.

Leverage amplifies conduct when gross advantage scales faster than influence, detection, execution, and sanction costs. That condition can hold, but it is not automatic and cannot be summarized by an event-class slogan or a universal leverage threshold. The appropriate design response is functional: control position capacity, preserve attributable order data, restrict influential participants, govern resolution separately from trading, and maintain explicit loss-bearing capital.

The main scientific conclusion is therefore not that leveraged event markets are inherently impermissible or inherently safe. It is that their conduct risks are separable, measurable at different evidence layers, and governed by different mechanisms. A credible venue must demonstrate control over all of them.

\appendix
\section{Additional comparative statics}
Let $V(N)=N\Delta q-K-D(N)-X(N)$.
\begin{lemma}[Interior optimum]
If $D+X$ is strictly convex, continuously differentiable, and $(D'+X')(0)<\Delta q<\lim_{N\to\bar N}(D'+X')(N)$, then $V$ has a unique interior maximizer $N^*$ satisfying
\[
D'(N^*)+X'(N^*)=\Delta q.
\]
\end{lemma}
\begin{proof}
Strict convexity of $D+X$ makes $V$ strictly concave. The derivative changes sign by the endpoint conditions, yielding a unique root and maximizer.
\end{proof}

\begin{corollary}[Effect of a binding position limit]
If an enforceable aggregate limit $\ell<N^*$ applies, the maximal conduct value is $V(\ell)<V(N^*)$.
\end{corollary}

The corollary concerns enforceable aggregate exposure. Per-wallet limits need not bind when identities can be fragmented or exposure can be replicated across venues and related contracts.

\section{Regulatory-source discipline}
The regulatory section follows four rules: use primary sources where available; date every snapshot; distinguish enacted rules from proposals, advisories, litigation positions, and no-action relief; and avoid inferring authorization from the absence of a prohibition.

\begin{table}[ht]
\centering\small
\caption{United States primary-source sequence used in this revision.}
\begin{tabularx}{\textwidth}{@{}L{2.4cm}L{3.4cm}Y@{}}
\toprule
Date & Source type & Relevance\\
\midrule
4 Feb 2026 & CFTC withdrawal & 2024 event-contract proposal withdrawn\\
25 Feb 2026 & Enforcement advisory & Non-public information and prohibited influence cases\\
12 Mar 2026 & ANPRM and market advisory & New rulemaking inquiry and DCM obligations\\
10 Jun 2026 & NPRM & Enumerated-activity/public-interest framework proposal\\
25 Jun 2026 & Reporting NPRM & Proposed reporting treatment for certain fully collateralized event contracts\\
\bottomrule
\end{tabularx}
\end{table}

\section{Notation}
\begin{tabularx}{\textwidth}{@{}L{2.2cm}Y@{}}
\toprule
Symbol & Meaning\\
\midrule
$Y$ & Final binary payout\\
$p$ & Entry price\\
$C$ & Trader capital\\
$N$ & Gross exposure\\
$L=N/C$ & Leverage\\
$a$ & Conduct action\\
$K(a)$ & Direct action cost\\
$D(a,N)$ & Detection and sanction cost\\
$X(a,N)$ & Execution and price-impact cost\\
$\Delta q(a)$ & Change in favorable-outcome probability caused by action\\
$e$ & Informational edge relative to price\\
\bottomrule
\end{tabularx}

\printbibliography[heading=bibintoc,title={References}]

@online{cftc_prediction_2026,author={{Commodity Futures Trading Commission}},title={Understanding Prediction Markets and Event Contracts},year={2026},url={https://www.cftc.gov/LearnandProtect/PredictionMarkets},urldate={2026-07-19}}

@online{cftc_anprm_2026,author={{Commodity Futures Trading Commission}},title={Prediction Markets: Advance Notice of Proposed Rulemaking},year={2026},url={https://www.cftc.gov/LawRegulation/FederalRegister/proposedrules/2026-05105.html},urldate={2026-07-19}}

@online{cftc_withdrawal_2026,author={{Commodity Futures Trading Commission}},title={CFTC Withdraws Event Contracts Rule Proposal and Staff Sports Event Contracts Advisory},year={2026},url={https://www.cftc.gov/PressRoom/PressReleases/9179-26},urldate={2026-07-19}}

@online{cftc_nprm_2026,author={{Commodity Futures Trading Commission}},title={CFTC Seeks Public Comment on Notice of Proposed Rulemaking Concerning Event Contracts Involving Enumerated Activities},year={2026},url={https://www.cftc.gov/PressRoom/PressReleases/9249-26},urldate={2026-07-19}}

@online{cftc_reporting_2026,author={{Commodity Futures Trading Commission}},title={CFTC Seeks Public Comment on Data Reporting Requirements for Certain Event Contracts},year={2026},url={https://www.cftc.gov/PressRoom/PressReleases/9261-26},urldate={2026-07-19}}

@online{cftc_enforcement_2026,author={{Commodity Futures Trading Commission}},title={CFTC Enforcement Division Issues Prediction Markets Advisory},year={2026},url={https://www.cftc.gov/PressRoom/PressReleases/9185-26},urldate={2026-07-19}}

@online{cftc_market_advisory_2026,author={{Commodity Futures Trading Commission}},title={CFTC Staff Issues Prediction Markets Advisory},year={2026},url={https://www.cftc.gov/PressRoom/PressReleases/9193-26},urldate={2026-07-19}}
\end{document}